\documentclass[11pt,a4paper,reqno]{amsart}%
\usepackage{amsthm,amsmath,amsfonts,amssymb,amsxtra,appendix,bookmark,euscript,delarray,dsfont,bm,mathrsfs}
\usepackage[latin1]{inputenc}
\usepackage{color} 

\usepackage[version=4]{mhchem}

\usepackage{graphicx}

\theoremstyle{plain}
\newtheorem{theorem}{Theorem}

\theoremstyle{remark}
\newtheorem{remark}[theorem]{Remark}
\newtheorem{definition}[theorem]{Definition}

\renewcommand\phi{\varphi}

\newcommand{\eps}{\epsilon}

\renewcommand{\epsilon}{\varepsilon}

\renewcommand{\ge}{\geqslant}
\renewcommand{\le}{\leqslant}
\renewcommand{\geq}{\geqslant}
\renewcommand{\leq}{\leqslant}

\renewcommand{\tilde}{\widetilde}

\allowdisplaybreaks

\title[A  reaction network approach to  acoustic wave turbulence]{A  reaction network approach to the theory of acoustic wave turbulence}

\author[M.-B. Tran]{Minh-Binh Tran}
\address{Department of Mathematics, Southern Methodist University} 
\email{minhbinht@mail.smu.edu}
\author[G. Craciun]{Gheorghe Craciun}
\address{Department of Mathematics and Department of Biomolecular Chemistry, University of Wisconsin-Madison} 
\email{craciun@wisc.edu}
\author[L. M. Smith]{Leslie M. Smith}
\address{Department of Mathematics and Department of Engineering Physics \\
University of Wisconsin-Madison} 
\email{lsmith@math.wisc.edu}
\author[S. Boldyrev]{Stanislav Boldyrev}
\address{Department of Physics,
University of Wisconsin-Madison} 
\email{boldyrev@wisc.edu}

\begin{document}

\maketitle
\begin{abstract} We propose a new approach to study the long time dynamics of the wave kinetic equation in the statistical description  of  acoustic  turbulence. The approach is based on rewriting the discrete version of the wave kinetic equation in the form of a chemical reaction network, then employing techniques used to study the Global Attractor Conjecture to investigate the long time dynamics of the newly obtained chemical system. We show that the solution of the chemical system converges to an equilibrium exponentially in time. In addition, a resonance broadening modification of the acoustic wave kinetic equation is also studied with the same technique. 
For the near-resonance equation, 
if the resonance broadening frequency is larger than a threshold, the solution of the system goes to infinity as time evolves.
\end{abstract}{\bf Keywords:} kinetic theory, Boltzmann equation, weak turbulence, dynamical systems,  rate of convergence to equilibrium, global attractor conjecture, mass-action kinetics, power law
systems, biochemical networks, acoustic waves.
{
\tableofcontents

\section{Introduction}\label{Intro}

Describing the behavior of a spatially homogeneous field of random, weakly interacting dispersive waves, the theory of wave turbulence has been very successful  
in explaining the processes of spectral energy transfer in several 
areas of modern science, 
such as oceanography, plasmas, planetary waves, acoustic turbulence etc. \cite{zakharov2012kolmogorov}.  The central ingredient of the theory is the derivation of the kinetic-wave equations, which describe the spectral energy transfer via $n$-wave resonant processes, and which are in one to one correspondence with the spectral moments \cite{newell2011wave}. Without making a priori assumptions on the statistics of the processes, these equations are 
closed. An important feature of these equations is that their exact equilibrium solutions have constant spectral fluxes of one of the conserved energy densities and the number density. 

The derivation of kinetic equations to describe how weakly interacting waves share their energies in anharmonic crystal lattices in solid state physics go back to Peierls, in the early 30's \cite{Peierls:1993:BRK,Peierls:1960:QTS}. Indeed, to our knowledge, Peirels' model is the first wave turbulence kinetic equation derived. The modern theory has been then developed  based on the pioneering works of Hasselmann \cite{Hasselmann:1962:OTN1}, Benney and Saffmann \cite{benney1966nonlinear}, Kadomtsev \cite{kadomtsev1988collective}, Zakharov \cite{zakharov2012kolmogorov}, Benney and Newell \cite{benney1967propagation}. A great breakthough in the theory is the discovery of the Zakharov-Kolmogorov solution \cite{zakharov1965weak,zakharov1967instability} by using the scaling symmetries of the dispersion relation and the coupling coefficient via the Zakharov transformation \cite{zakharov1965weak,zakharov1967instability,zakharov2012kolmogorov}.

In \cite{l1997statistical}, the authors develop expressions for the nonlinear wave damping and frequency correction of a field of random,
spatially homogeneous, acoustic waves. They derived the 3-wave kinetic equation of acoustic waves, describing the evolution of the density distribution function $f$ of the waves, in which the distribution $f(t,p)$ is a function of time $t$ and wave number $p$. If we denote 
$$f_1=f(t,p_1), f_2=f(t,p_2), f_3=f(t,p_3),$$
then $f_1$ satisfies
\begin{eqnarray}
\label{QuantumBoltzmann1}
\frac{\partial f_1}{\partial t}&=&Q[f_1],
\end{eqnarray}
where
\begin{equation}\begin{aligned}
Q[f_1]:=\ &\ \int_{\mathbb{R}^{6}}\mathcal V_{p_1, p_2, p_3}\delta(p_1-p_2-p_3)\delta(\omega_{p_1}-\omega_{p_2}-\omega_{p_3})[f_2f_3-f_1(f_2+f_3)]dp_2dp_3\\
\ & \ -2\int_{\mathbb{R}^{6}}\mathcal V_{p_1, p_2, p_3}\delta(p_2-p_1-p_3)\delta(\omega_{p_2}-\omega_{p_1}-\omega_{p_3})[f_1f_3-f_2(f_1+f_3)]dp_2dp_3.\end{aligned}
\end{equation}
The collision kernels
 $\mathcal V_{p_1, p_2, p_3}\geq 0$ are radially symmetric, and symmetric with respect to the permutation of $p_1$, $p_2$, $p_3,$:
\begin{align*}
&~~\mathcal V_{p_1, p_2, p_3}=\lambda|p_1||p_2||p_3|,
\end{align*}
where $|p|$ denotes the length of the vector $p$, and $\lambda$ is a positive constant that we can assume without loss of generality to be $1$.   In the case of acoustic waves, the dispersion relation $\omega_p=\omega (p)$ is given by the phonon dispersion law:
\begin{equation}
\omega (p)=|p|. \label{E3a}
\end{equation}

In this paper, besides the exact resonance equation \eqref{QuantumBoltzmann1}, we also consider the following near-resonance turbulence kinetic equation for acoustic waves
\begin{eqnarray}
\label{NRQuantumBoltzmann1}
\frac{\partial f_1}{\partial t}&=&Q^{NR}[f_1],
\end{eqnarray}
where
\begin{equation}\begin{aligned}
Q^{NR}[f_1]:=\ & \ \int_{\{p_1=p_2+p_3,|\omega_{p_1}-\omega_{p_2}-\omega_{p_3}|\le \Lambda\}}\mathcal V_{p_1, p_2, p_3}[f_2f_3-f_1(f_2+f_3)]dp_2dp_3\\
\ & \ -2\int_{\{p_2=p_1+p_3,|\omega_{p_2}-\omega_{p_1}-\omega_{p_3}|\le \Lambda\}}\mathcal V_{p_1, p_2, p_3}[f_1f_3-f_2(f_1+f_3)]dp_2dp_3.
\end{aligned}
\end{equation} 
Indeed, equation \eqref{QuantumBoltzmann1} has an explicit expression for the Kolmogorov-type spectrum
of acoustic turbulence
which is just a stationary solution of the equation  \cite{l1997statistical}. Such a solution has not been found for the resonance broadening equation \eqref{NRQuantumBoltzmann1}. In the current paper, we are interested in a different mathematical question on the time-dependent solutions rather than the stationary solution. 

In general, the 3-wave kinetic equation plays an important role in the theory of weak turbulence, and has been rigorously studied in \cite{nguyen2017quantum} for capillary waves, in  \cite{AlonsoGambaBinh,CraciunBinh,EscobedoBinh,GambaSmithBinh} for the phonon interactions in anharmonic crystal lattices, in \cite{SofferBinh3} for acoustic waves, and in \cite{GambaSmithBinh} for stratified  flows in the ocean. 

It is the goal of our work to make the connection between equations \eqref{QuantumBoltzmann1}, \eqref{NRQuantumBoltzmann1} and chemical reaction systems. Understanding the qualitative behavior of deterministically modeled chemical reaction systems has been a subject of great interests during the past four decades \cite{Yu_Craciun_survey}. The main questions  include the existence of positive equilibria, stability properties of equilibria, and the non-extinction, or persistence, of species \cite{Anderson:2001:APG,MR2734052,AngeliLeenheerSontag:2011:PRF,MR2604624,MR2561288,Feinberg:1972:CBI,Feinberg:1995:TEA,MR3199409, Horn:1974:TDO,Yu_Craciun_survey}.  Used to describe an important class of chemical kinetics, toric dynamical systems or complex-balanced systems  \cite{MR2561288,HornJackson:1972:GMA} are among the most important models in chemical reaction network. First introduced by Boltzmann  \cite{Boltzmann} for modeling collisions in kinetic gas theory, the complex-balanced condition was used by Horn and Jackson \cite{Feinberglecture,Gunawardena,Horn:1972:NAS,HornJackson:1972:GMA,Yu_Craciun_survey} to show  that a complex-balanced system has a unique locally stable equilibrium within each linear invariant subspace. Later on,   the name ``toric dynamical system'' was proposed in \cite{MR2561288} to underline the tight connection to the algebraic study of toric varieties. The Global Attractor Conjecture, the most important problem for toric dynamic systems, states that the complex balanced equilibrium of each system is also a globally attracting point within each linear invariant subspace. The Global Attractor Conjecture has been proved in \cite{CraciunNazarovPantea:2013:PAP}  for small dimensional systems, and a solution has been proposed in \cite{Craciun:2015:TDI} for the general case. 

In this paper, {\it we discover, for the first time, the connection between the wave kinetic equation \eqref{QuantumBoltzmann1} and chemical reaction networks}. We prove that the discrete version of \eqref{QuantumBoltzmann1} can be associated with a chemical reaction network which takes the form
\begin{eqnarray*}
&&\ce{A_{k_2} + A_{k_3} -> A_{k_1}}\\
&&\ce{A_{k_2} + A_{k_1} -> 2A_{k_2} + A_{k_3}}
\end{eqnarray*} 
and will be described in detail in Section \ref{Sec:ChemC12}. 
As a consequence, techniques that have been used to study the  Global Attractor Conjecture in chemical reaction network theory can be applied to study the long time behavior of the wave kinetic equation \eqref{QuantumBoltzmann1}. We prove that as time evolves, the solution of the discrete version of \eqref{QuantumBoltzmann1} converges to a steady state exponentially in time.

The Dirac-delta functions in \eqref{QuantumBoltzmann1} imply that the 
spectral energy transfer occurs on the resonant manifold, which is a set of wave vectors $p$, $p_1$, $p_2$ satisfying
\begin{equation}
\label{ResonantManifold}
p=p_1+p_2,\ \ \ \ \ \ \omega_p=\omega_{p_1}+\omega_{p_2}.
\end{equation}
However, in related systems, it is shown that exact resonances with $\omega_p=\omega_{p_1}+\omega_{p_2}$ do not capture some important physical effects
\cite{bartello1995geostrophic,gill1967resonant,greenspan1969non,phillips1968interaction, warn1986statistical}. Therefore one needs to include more physics by  adding near-resonant interactions  
\cite{boldyrev2012residual,connaughton2001discreteness,huang2000,lee2007formation,l1997statistical,lvov2012resonant,newell1969rossby,smith2001,smith2005near,smith1999transfer,smith2002generation}, 
which satisfy
\begin{equation}
\label{NearResonantManifold}
p=p_1+p_2,\ \ \ \ |\omega_p-\omega_{p_1}-\omega_{p_2}|<\Lambda, 
\end{equation}
where $\Lambda$ is the resonance broadening frequency.

By extending the chemical reaction network approach used to study the discrete versions of \eqref{QuantumBoltzmann1}  and \eqref{NRQuantumBoltzmann1}, we prove that:
\begin{itemize}
\item There exists a positive constant $\Lambda_0$ such that when $0\le \Lambda<\Lambda_0$, the solution of the discrete version of \eqref{NRQuantumBoltzmann1} 
converges to an unique equilibrium exponential having the form $|p|^{-1}$ in time, similar to the exact resonance case, which
is the same as the equilibrium solution of the exact resonance case. Note that when $\Lambda=0$ the equation \eqref{NRQuantumBoltzmann1} becomes  \eqref{QuantumBoltzmann1}.
\item There exists a positive constant $\Lambda_1$ such that when $\Lambda\ge \Lambda_1$, the solution of the discrete version of \eqref{NRQuantumBoltzmann1}  exits any compact set  as $t$ tends to infinity.
\end{itemize}

Besides the 3-wave kinetic equation, the rigorous study of 4-wave kinetic equation is a very important subject (see \cite{buckmaster2016analysis,buckmaster2016effective,faou2016weakly,germain2016continuous,germain2017optimal}   and references therein).
Finally, we note that wave-turbulence kinetic equations have very similar form with the quantum Boltzmann equations, used to describe the evolution of diluted bose gases at high and low temperature. We refer to the book \cite{PomeauBinh} or the papers \cite{JinBinh,ToanBinh,tran2019boltzmann,ReichlTran,saint2004kinetic,SofferBinh2,SofferBinh1} and the references therein for more discussions on the latter topic.

The plan of our paper is the following:
\begin{itemize}
\item In section \ref{Sec:Q12}, we show that the discrete version of equation \eqref{QuantumBoltzmann1} can be rewritten as a chemical reaction network. By using an approach inspired by the theory of toric dynamical system, we prove in Theorem \ref{TheoremQ12} that the solution of the discrete version of \eqref{QuantumBoltzmann1} converges to the equilibrium exponentially in time.
\item In section \ref{NRSec:Q12}, we generalize Theorem \ref{TheoremQ12} to  the near resonance case \eqref{NRQuantumBoltzmann1}. We prove that depending on the resonance broadening frequency $\Lambda$, the solution of the discrete version of the equation may   converge to the equilibrium  exponentially or go to infinity as time evolves. 
\end{itemize}
\section{A reaction network approach to the exact resonance equation}\label{Sec:Q12}
\subsection{The dynamical system associated with the exact resonance equation}
Let us consider the discrete version of  \eqref{QuantumBoltzmann1}, which is described below.
\\\\ Let $\mathcal{L}_R$ denote the lattice of integer points 
$$\mathcal{L}_R=\{p\in\mathbb{Z}^3, |p|<R\}.$$
The discrete version of the quantum Boltzmann equation \eqref{QuantumBoltzmann1} reads 
\begin{equation}\label{DiscreteQuantum}\begin{aligned}
\dot{f}_{p_1}=
&\sum_{\substack{p_2,p_3\in\mathcal{L}_R,\\ p_1-p_2-p_3=0, \\ \omega(p_1)-\omega(p_2)-\omega(p_3)=0}}\mathcal{V}_{p_1,p_2,p_3}\left[f_{p_2}f_{p_3}-f_{p_1}(f_{p_2}+f_{p_3})\right]\\
&\quad -2\sum_{\substack{p_2,p_3\in\mathcal{L}_R,\\ p_1+p_2-p_3=0,\\ \omega(p_1)+\omega(p_2)-\omega(p_3)=0}}\mathcal{V}_{p_1,p_2,p_3}\left[f_{p_1}f_{p_2}-f_{p_3}(f_{p_1}+f_{p_2})\right],~~\end{aligned}
\end{equation}
for all $p_1$ in $\mathcal{L}_R$, where  $\omega(p)$ is defined in \eqref{E3a}.
\subsection{Decoupling the exact resonance equation}
 Observe  that when $p_1=0$, $\mathcal{V}_{p_1,p_2,p_3}$ is also $0$, and therefore, 
\begin{eqnarray}\label{DiscreteQuantumIndex0}
\dot{f}_0=0,
\end{eqnarray}
which implies that $f_0(t)$ is a constant for all $t\ge0$. Moreover, $f_{p_1}$ does not depend on $f_0$ for all $p_1\ne 0$. Therefore, without loss of generality, we can suppose that $f_0(0)=0$, which leads to $f_0(t)=0$ for all $t$.

Taking into account the fact  $\omega(p)=c|p|$, note that if $p_1,p_2,p_3\in\mathcal{L}_R$ are different from $0$, and for $p_3=p_1+p_2$ and $|p_3|=|p_1|+|p_2|$ (like in the second sum of \eqref{DiscreteQuantum}), then $p_1,p_2,p_3$ must be collinear and on the same side of the origin. Therefore, we infer that there exists a vector $P$  and  $k_1$, $k_2$, $k_3>0$, $k_1,k_2,k_3\in\mathbb{Z}$ such that 
$$p_1=k_1 P;~~~p_2=k_2 P;~~~p_3=k_3 P,~~~k_1+k_2=k_3.$$  
Since $\mathcal{L}_R$ is bounded, it follows that $k_1,k_2,k_3$ belong to a finite set of integer indices $\mathbb{I}=\{1,\dots,I\}$.
Arguing similarly for the first sum in \eqref{DiscreteQuantum}, we deduce that \eqref{DiscreteQuantum} is equivalent to the following system for $k_1\in\mathbb{I}$
\begin{equation}\label{DiscreteQuantum1DVersion2}\begin{aligned}
\dot{f}_{Pk_1}
=&\quad\sum_{\substack{k_2,k_3\in
\mathbb{I},\\ k_1-k_2-k_3=0}}\mathcal{V}_{Pk_1,Pk_2,Pk_3}\left[f_{Pk_2}f_{Pk_3}-f_{Pk_1}(f_{Pk_2}+f_{Pk_3})\right]\\
&\quad-{2}\sum_{\substack{k_2,k_3\in
\mathbb{I},\\ k_1+k_2-k_3=0}}\mathcal{V}_{Pk_1,Pk_2,Pk_3}\left[f_{Pk_1}f_{Pk_2}-f_{Pk_3}(f_{Pk_1}+f_{Pk_2})\right].
\end{aligned}
\end{equation}
Note that the system of equations \eqref{DiscreteQuantum1DVersion2} shows a {\it decoupling} of the system of equations \eqref{DiscreteQuantum} along a ray $\{kP_0\}$ with $k>0$ (see Figure 1). As a consequence, it is sufficient to study the system of equations \eqref{DiscreteQuantum1DVersion2} for a fixed value of $P_0$, instead of the system of equations \eqref{DiscreteQuantum}.

If we denote $f_{k_1P_0}$ by $f_{k_1}$ (with $k_1\in \mathbb{I}$) and $\mathcal{V}^{12}_{k_1P_0,k_2P_0,k_3P_0}$  by $\mathcal{V}^{12}_{k_1,k_2,k_3}$, with an abuse of notation, we obtain the following new system for the ray $\{k_1P_0|k_1>0\}$:
\begin{equation}\label{DiscreteQuantum1D}\begin{aligned}
\dot{f}_{k_1}=&\quad\sum_{\substack{k_2,k_3\in
\mathbb{I},\\ k_1=k_2+k_3}}\mathcal{V}_{k_1,k_2,k_3}[f_{k_2}f_{k_3}-f_{k_1}(f_{k_2}+f_{k_3})]
\\
&\quad-2\sum_{\substack{k_2,k_3\in
\mathbb{I},\\ k_1+k_2=k_3}}\mathcal{V}_{k_1,k_2,k_3}[f_{k_1}f_{k_2}-f_{k_3}(f_{k_1}+f_{k_2})], ~\forall k_1\in\mathbb{I}.\end{aligned}
\end{equation}
The {\it conservation of energy} then follows
\begin{equation}\label{MassConservation1}
\sum_{k=1}^I k\dot{f_k}=0,
\end{equation}
or equivalently
\begin{equation}\label{MassConservation2}
\sum_{k=1}^I k{f_k}=\mbox{const}.
\end{equation}
These conservation relations also follow easily from the rewriting of equations \eqref{DiscreteQuantum1D} as a mass-action system, as described in the next section.
Also by abuse of notation, we denote this discrete version of $Q$ by
\begin{equation}\label{DiscreteC12}
\begin{aligned}
\mathcal{Q}[f_{k_1}]:= &~~\sum_{k_2+k_3=k_1}\mathcal{V}_{k_1,k_2,k_3}[f_{k_2}f_{k_3}-f_{k_1}(f_{k_2}+f_{k_3})]\\
&~~-2\sum_{k_1+k_2=k_3}\mathcal{V}_{k_3,k_1,k_2}[f_{k_1}f_{k_2}-f_{k_3}(f_{k_1}+f_{k_1})].
\end{aligned}
\end{equation} 

\subsection{The chemical reaction network associated with the exact resonance equation}\label{Sec:ChemC12}
For $X\in\mathbb{R}_{>0}^n$ and $\alpha\in \mathbb{R}_{\geq 0}^n$, we denote by $X^\alpha$ the monomial $\Pi_{i=1}^n X_i^{\alpha_i}$.
\begin{definition} Consider a chemical reaction of the form
$$\ce{\alpha_1 A_1  +  \alpha_2 A_2  +  ...    +  \alpha_n A_n ->[\mathcal{V}] \beta_1 A_1  +  \beta_2 A_2  +  ...  +  \beta_n A_n},$$ 
where $\mathcal{V}$ is a positive parameter, called  reaction rate constant. Then the mass-action dynamical system generated by this reaction is
\begin{equation}\label{ExampleEquation}
\dot{X}= \mathcal{V} X^\alpha (\beta-\alpha),
\end{equation}
where $\alpha=(\alpha_1,\cdots,\alpha_n)^T$, $\beta=(\beta_1,\cdots,\beta_n)^T$, $\alpha_i,\beta_i \ge 0$ and  $X=(X_1,\cdots,X_n)^T$, in which $X_i$ is the concentration of the chemical species $A_i$. For the case of a network that contains several reactions
$$\ce{\alpha_1^j A_1^j  +  \alpha_2^j A_2^j  +  ...   +   \alpha_n^j A_n^j ->[\mathcal{V}_j] \beta_1 A_1^j  +  \beta_2^j A_2^j  +  ...   +  \beta_n^j A_n^j},$$ 
for $1\le j\le m$, its associated mass-action dynamical system is given by
\begin{equation}\label{DynSys}
\dot{X}=\sum_{j=1}^m \mathcal{V}_jX^{\alpha^j}(\beta^j-\alpha^j).
\end{equation}
\end{definition}

In this section, we will show  that the system \eqref{DiscreteQuantum1D} has the  form \eqref{DynSys} for a well-chosen set of reactions. 
We will derive  the system \eqref{DiscreteQuantum1D} from the network of chemical reactions of the form:
\begin{eqnarray}\label{BioChemEq01}
&&\ce{A_{k_2} + A_{k_3} -> A_{k_1}}\\\label{BioChemEq02}
&&\ce{A_{k_2} + A_{k_1} -> 2A_{k_2} + A_{k_3}},
\end{eqnarray}
for all $k_1,k_2,k_3$ in $\mathbb{I}$ such that $k_2+k_3=k_1$. If we denote by $F_k$ the concentration of the species $X_k$, we will show that, {\it for appropriate choices of the reaction rate constants} in \eqref{BioChemEq01} and \eqref{BioChemEq02}, the differential equations satisfied by $F_k$ according the mass-action kinetics are exactly the same as \eqref{DiscreteQuantum1D}.\\

In order to describe the connection between the mass-action system given by reactions of the form \eqref{BioChemEq01}-\eqref{BioChemEq02} and our system \eqref{DiscreteQuantum1D}, we need to consider several cases.

\bigskip

{\it Case 1:} For $k_2+k_3=k_1$, $k_2\ne k_3$, $k_1,k_2,k_3\in\mathbb{I}$, we consider 
\begin{eqnarray}\label{BioChemEq1a}
&&\ce{A_{k_2} + A_{k_3} ->[2\mathcal{V}_{k_1,k_2,k_3}] A_{k_1}}\\\label{BioChemEq1b}
&&\ce{A_{k_2} + A_{k_1} ->[2\mathcal{V}_{k_1,k_2,k_3}] 2A_{k_2} + A_{k_3}},
\end{eqnarray}
and for the  reaction \eqref{BioChemEq1a}-\eqref{BioChemEq1b}, we choose the reaction rate constants of the three reactions $A_{k_2}+A_{k_3}\to A_{k_1}$, $A_{k_2}+A_{k_1}\to 2A_{k_2}+A_{k_3}$ to be $2\mathcal{V}_{k_1,k_2,k_3}$.
For example,  consider the  reaction \eqref{BioChemEq1a}:
in this reaction, $A_{k_1}$ is created from $A_{k_2}+A_{k_3}$ with the rate $2\mathcal{V}_{k_1,k_2,k_3}F_{k_2}F_{k_3}$. Therefore, the rate of change of the species $A_{k_1}$ due to this reaction is $ 2\mathcal{V}_{k_1,k_2,k_3}F_{k_2}F_{k_3}$.
For the  reaction  \eqref{BioChemEq1b},
 $A_{k_1}$ is lost with the rate $-2\mathcal{V}_{k_1,k_2,k_3}F_{k_2}F_{k_1}$ to create $2A_{k_2}+A_{k_3}$. Therefore the rate of change of the species $A_{k_1}$ due to this reaction is $ -2\mathcal{V}_{k_1,k_2,k_3}F_{k_2}F_{k_1}$. By exchanging the roles of $A_{k_2}$ and $A_{k_3}$ in \eqref{BioChemEq1b}, we obtain the rate $ -2\mathcal{V}_{k_1,k_2,k_3}[F_{k_2}F_{k_1}+F_{k_3}F_{k_1}]$.
Therefore, the total rate of change of $A_{k_1}$ due to the reactions in \eqref{BioChemEq1a}-\eqref{BioChemEq1b} is
\begin{equation}\label{FirstChangeRate}
\begin{aligned}
&2\mathcal{V}_{k_1,k_2,k_3}[F_{k_2}F_{k_3}- F_{k_2}F_{k_1}-F_{k_3}F_{k_1}].
\end{aligned}
\end{equation}

\bigskip

{\it Case 2:}  For $k_2+k_3=k_1$, $k_2=k_3$,  we obtain $2k_2=k_1$, $k_1,k_2\in\mathbb{I}$, and we consider
\begin{eqnarray}\label{BioChemEq1a2}
&&\ce{2A_{k_2} ->[\mathcal{V}_{k_1,k_2,k_3}] A_{k_1}}\\\label{BioChemEq1b2}
&&\ce{A_{k_2} + A_{k_1} ->[2\mathcal{V}_{k_1,k_2,k_3}] 3A_{k_2} }.
\end{eqnarray}
We choose the reaction rate constant of $2A_{k_2}\to A_{k_1}$  to be $\mathcal{V}_{k_1,k_2,k_3}$. Also, we choose the reaction rate constant of $A_{k_2}+A_{k_1}\to 3A_{k_2}$ to be $2\mathcal{V}_{k_1,k_2,k_3}$.
Consider the first reaction \eqref{BioChemEq1a2}:
In this reaction, $A_{k_1}$ is created from $2A_{k_2}$ with the rate $\mathcal{V}_{k_1,k_2,k_2}F_{k_2}^2$. The rate of change of the species $A_{k_1}$ is $\mathcal{V}_{k_1,k_2,k_2}F_{k_2}^2$. 
For the second reaction  \eqref{BioChemEq1b2}:
 $A_{k_1}$ is lost with the rate $-2\mathcal{V}_{k_1,k_2,k_2}F_{k_2}F_{k_1}$ to create $3A_{k_2}$. 
As a result, the rate of change of $A_{k_1}$ due to the reactions  \eqref{BioChemEq1a2}-\eqref{BioChemEq1b2} is
\begin{equation}\label{FirstChangeRate2}
\begin{aligned}
&\mathcal{V}_{k_1,k_2,k_3}[F_{k_2}^2- 2F_{k_2}F_{k_1}].
\end{aligned}
\end{equation}

In order to compute the total  rate of change of $A_{k_1}$, we need the combination of \eqref{BioChemEq1a}-\eqref{BioChemEq1b}, \eqref{BioChemEq1a2}-\eqref{BioChemEq1b2} and 
\begin{eqnarray}\label{BioChemEq1a3}
&&\ce{A_{k_1} + A_{k_3} ->[2\mathcal{V}_{k_2,k_1,k_3}] A_{k_2}}\\\label{BioChemEq1b3}
&&\ce{A_{k_2} + A_{k_1} ->[2\mathcal{V}_{k_2,k_1,k_3}] 2A_{k_1} + A_{k_3}},
\end{eqnarray}
which, by  \eqref{FirstChangeRate}, \eqref{FirstChangeRate2}, implies

\begin{equation}\label{TotalChangeRate}
\begin{aligned}
\dot{F}_{k_1}\ =
&~~\sum_{k_2+k_3=k_1, k_2< k_3}2\mathcal{V}_{k_1,k_2,k_3}[F_{k_2}F_{k_3}-F_{k_1}(F_{k_2}+F_{k_3})]\\
&~~+\sum_{2k_2=k_1}\mathcal{V}_{k_1,k_2,k_2}[F_{k_2}F_{k_2}-F_{k_1}(F_{k_2}+F_{k_2})]\\
&~~-\sum_{k_1+k_3=k_2}2\mathcal{V}_{k_2,k_1,k_3}[F_{k_1}F_{k_3}-F_{k_2}(F_{k_1}+F_{k_3})],
\end{aligned}
\end{equation} 
which can be written as
\begin{equation}\label{TotalChangeRateEquation}
\begin{aligned}
\dot{F}_{k_1}\ =&~~\sum_{k_2+k_3=k_1}\mathcal{V}_{k_1,k_2,k_3}[F_{k_2}F_{k_3}-F_{k_1}(F_{k_2}+F_{k_3})]\\
&~~-2\sum_{k_1+k_3=k_2}\mathcal{V}_{k_2,k_1,k_3}[F_{k_1}F_{k_3}-F_{k_2}(F_{k_1}+F_{k_3})].
\end{aligned}
\end{equation} 
Equation (\ref{TotalChangeRateEquation}) shows that the system of differential equations satisfied by the concentrations $F_k$ is exactly the same as the system of differential equations \eqref{DiscreteQuantum1D} satisfied by the densities $f_k$.

\subsection{A Lyapunov function inspired by the associated reaction network}\label{Sec:LyapC12}

{\em The  observation above is very interesting because it shows a very strong connection between our system \eqref{DiscreteQuantum1D} and the  reaction network model \eqref{BioChemEq01}-\eqref{BioChemEq02}.}  As a consequence, the techniques developed to study the  Global Attractor Conjecture \cite{Craciun:2015:TDI} in  reaction network theory can be applied to study \eqref{DiscreteQuantum1D}. One of the key ingredients in proving the convergence to global attractors of detailed balanced or complex balanced reaction network models is that these networks have a specific type of Lyapunov functions. Then our system of interest \eqref{DiscreteQuantum1D} may also have a similar Lyapunov function. By using a change of variable similar to the approach in \cite{CraciunBinh}, we can create a global Lyapunov function (related to  Boltzmann's original H-theorem calculations), as follows.

To illustrate this idea in a very simple way, we select in the  system above  three values $k_1, k_2, k_3$ such that $k_1 + k_2 = k_3$, and suppose $\mathcal{V}_{k_1,k_2,k_3}=1$. Then this simplified ``sub-system" can be rewritten as
\begin{equation}
\begin{aligned}
\frac{d}{dt}\left( {\begin{array}{cc}
  {F_{k_1}}\\        {F_{k_2}} \\     {F_{k_3}}  \end{array} } \right) \ = \ (F_{k_1}F_{k_3}+F_{k_2}F_{k_3}-F_{k_1}F_{k_2})\left(\left( {\begin{array}{cc}
  1\\        1 \\     0  \end{array} } \right)-\left( {\begin{array}{cc}
  0\\        0 \\     1  \end{array} } \right) \right),
\end{aligned}
\end{equation}
which can be developed into
\begin{equation}
\begin{aligned}
\frac{d}{dt}\left( {\begin{array}{cc}
  {F_{k_1}}\\        {F_{k_2}} \\     {F_{k_3}}  \end{array} } \right) \ = \ F_{k_1}F_{k_2}F_{k_3}\left(\frac{1}{F_{k_2}}+\frac{1}{F_{k_1}}-\frac{1}{F_{k_3}}\right)\left(\left( {\begin{array}{cc}
  1\\        1 \\     0  \end{array} } \right)-\left( {\begin{array}{cc}
  0\\        0 \\     1  \end{array} } \right)\right).
\end{aligned}
\end{equation}
Consider the change of variables
$$G_k=\exp\left(\frac{1}{F_k}\right), \ \ \ \ F_k=\frac{1}{\log(G_k)},$$
then  we have
$$\frac{1}{F_{k_1}}+\frac{1}{F_{k_2}}-\frac{1}{F_{k_3}} \ = \ \log(G_{k_1}G_{k_2})-\log(G_{k_3}).$$
We then get
\begin{equation}
\begin{aligned}
\frac{d}{dt}\left( {\begin{array}{cc}
  {G_{k_1}}\\        {G_{k_2}} \\     {G_{k_3}}  \end{array} } \right) \ = \ \mathrm{diag} \left( {\begin{array}{cc}
  -G_{k_1}\log^2(G_{k_1})\\        -G_{k_2}\log^2(G_{k_2}) \\     -G_{k_3}\log^2(G_{k_3})  \end{array} } \right)\frac{1}{\log G_{k_1}\log G_{k_2}\log G_{k_3}}[\log(G_{k_1}G_{k_2})-\log(G_{k_3})].
\end{aligned}
\end{equation}
By following the approach in \cite{CraciunBinh}, we construct a Lyapunov function for the system in variables $G_k$ given by
$$L(G)\  =\  \sum_{k=1}^I \log \log(G_k),$$
which, when translated back into the $F_k$ coordinates has a very simple form:
$$L(F)\  = \ \sum_{k=1}^I \log\left(\frac{1}{F_k}\right)\ = \ - \sum_{k=1}^I \log(F_k).$$ 
In the next section, we will see that this construction does indeed give rise to a global  strict Lyapunov function for our (full) dynamical system of interest.

\bigskip

\subsection{Convergence to equilibrium}

\begin{theorem}\label{TheoremQ12} 

For any positive initial condition, the solution $$f(t)=(f_p(t))_{p\in\mathcal{L}_R}$$ of the discrete wave turbulence equation \eqref{DiscreteQuantum} converges to an equilibrium state $f^*=(f_p^*)_{p\in\mathcal{L}_R}$. For each ray $\{kP_0\}_{k\geq 1}$ there exists a positive constant $\rho({P_0})$ such that if $p=kP_0$ then $$f^*_p=\frac{1}{{k\rho({P_0})}}.$$ Moreover, the solution $f(t)$ of \eqref{DiscreteQuantum} converges to $f^*$ exponentially fast, in the following sense: there exist positive constants $C_1$, $C_2$ such that $$\max_{p\in\mathcal{L}_R}|f_p(t)-f_p^*|<C_1e^{-C_2t}.$$
\end{theorem}
\begin{proof}
By using the decoupling and the change of variables discussed in the previous sections, for each ray $\{kP_0\}_{k\geq 1}$, we can reduce the study of $f$ to $F$, which satisfies \eqref{TotalChangeRateEquation}. 
\bigskip
{\\\it Step 1: The Lyapunov function.}
We define the  function
\begin{equation}\label{Lyapunov}
L(F)=-\sum_{k=1}^I \log (F_k),
\end{equation}
and we will show that $L$ is a Lyapunov function for  \eqref{TotalChangeRateEquation}.
Indeed, we have
\begin{equation}
\nabla L=\left( {\begin{array}{cc}
  -\frac{1}{F_1}\\       \cdots \\   -\frac{1}{F_I}  \end{array} } \right).
\end{equation}
By defining
\begin{equation}
e_k=\left( {\begin{array}{cc}
 0\\       \cdots \\ 1\\  \cdots \\  0  \end{array} } \right),
\end{equation}
in which the number $1$ stands for the $k$-th coordinate,
we then have
\begin{equation}
\begin{aligned}
& \ \dot{F}\cdot\nabla L = \\
 =  & \sum_{k_1+k_2=k}\mathcal{V}_{k_1,k_2,k}[F_{k_1}F_{k_2}-F_kF_{k_1}-F_kF_{k_2}][e_k-e_{k_1}-e_{k_2}]\cdot \left( {\begin{array}{cc}
  -\frac{1}{F_1}\\       \cdots \\   -\frac{1}{F_I}  \end{array} } \right) = \\
 =  & \sum_{k_1+k_2=k}\mathcal{V}_{k_1,k_2,k}F_kF_{k_1}F_{k_2}\left[\frac{1}{F_k}-\frac{1}{F_{k_1}}-\frac{1}{F_{k_2}}\right][e_k-e_{k_1}-e_{k_2}]\cdot \left( {\begin{array}{cc}
  -\frac{1}{F_1}\\       \cdots \\   -\frac{1}{F_I}  \end{array} } \right)\\
 =  &- \sum_{k_1+k_2=k}\mathcal{V}_{k_1,k_2,k}F_kF_{k_1}F_{k_2}\left[\frac{1}{F_k}-\frac{1}{F_{k_1}}-\frac{1}{F_{k_2}}\right]^2\\
 \leq  & \  0.
\end{aligned}
\end{equation}

In this case the vector $\beta-\alpha$ in  \eqref{ExampleEquation} takes the form $e_k-e_{k_1}-e_{k_2}$.
Also, note that the above inequality is strict unless
\begin{equation}\label{C12equilibrium}
\frac{1}{F_k}
=\frac{1}{F_{k_1}}+\frac{1}{F_{k_2}},\end{equation}
for all  $k=k_1+k_2.$
Equation \eqref{C12equilibrium}  implies that at equilibrium   $F_k^*=\frac{1}{\rho k}$, for some positive constant $\rho$. 
By the conservation relation
$$\sum_{k=1}^Ik{F_k}=\sum_{k=1}^Ik{F_k^*},$$
we deduce that $\rho$ is unique, i.e., \begin{equation}
F^*=\left( {\begin{array}{cc}
  \frac{1}{\rho} \\       \cdots \\  \frac{1}{ \rho I} \end{array} } \right)
\end{equation} is the only equilibrium point that satisfies the same conservation relation as the initial condition.
\\\bigskip\\
{\it Step 2: Differential inclusions and persistence.}

Since the Lyapunov function $L$ is infinite on the boundary,
we conclude that the system is persistent. 
Therefore, by using the existence of the globally defined strict Lyapunov function $L$, and the LaSalle invariance principle, it follows that all trajectories converge to the unique positive equilibrium $F^*$ that we discussed in~Step~1.
 \bigskip
{\\\\\it Step 3: Exponential rate of convergence.}

Despite the fact that the system \eqref{BioChemEq1a}-\eqref{BioChemEq1b} is irreversible, in this step, we will use methods that work well for reversible systems. To this end, we will introduce a change of variable technique, to convert \eqref{BioChemEq1a}-\eqref{BioChemEq1b} into the system for the reaction $A_{k_1}+A_{k_2}\leftrightarrow A_{k_3}$, which is indeed  reversible, but with reaction rate functions different from mass action kinetics. In other words, by converting the irreversible system \eqref{BioChemEq1a}-\eqref{BioChemEq1b} into a reversible one, we can then employ techniques originally designed for reversible systems.

Let us start with our  change of variable, which takes the following form $$G_k=\exp\left(\frac{1}{F_k}\right),$$
then 
$$F_k=\frac{1}{\ln(G_k)},$$
and 
$$
\begin{aligned}
F_{k_1}F_{k_2}-F_{k_1}F_{k_3}-F_{k_2}F_{k_3}  \ & \ =\frac{\ln(G_{k_1})-\ln(G_{k_2}G_{k_3})}{\ln(G_{k_1})\ln(G_{k_2})\ln(G_{k_3})}\\
\ & \ =\frac{\ln(G_{k_1})-\ln(G_{k_2}G_{k_3})}{\ln(G_{k_1})\ln(G_{k_2})\ln(G_{k_3})(G_{k_1}-G_{k_2}G_{k_3})}(G_{k_1}-G_{k_2}G_{k_3}).
\end{aligned}$$
Notice that $0<F_k<\infty$ and $1<G_k<\infty$. Moreover,
$$\frac{\ln(G_{k_1})-\ln(G_{k_2}G_{k_3})}{G_{k_1}-G_{k_2}G_{k_3}}>0, \forall k_1,k_2,k_3\in I.$$
We will now study $G_k$ instead of $F_k$. To do this, we convert the system  \eqref{TotalChangeRateEquation} into
\begin{equation}\label{DiscreteQuantumConverted}\begin{aligned}
&\frac{\dot{G}_{k_1}}{(\ln(G_{k_1}))^2}
= \tilde{\mathcal{Q}}[G](k_1)\\
:=&2\sum_{k_1+k_2=k_3}\mathcal{V}_{k_1,k_2,k_3}\frac{\ln(G_{k_3})-\ln(G_{k_2}G_{k_1})}{\ln(G_{k_1})\ln(G_{k_2})\ln(G_{k_3})(G_{k_3}-G_{k_2}G_{k_1})}(G_{k_3}-G_{k_1}G_{k_2})\\
+&\sum_{k_1=k_2+k_3}\mathcal{V}_{k_1,k_2,k_3}\frac{\ln(G_{k_1})-\ln(G_{k_2}G_{k_3})}{\ln(G_{k_1})\ln(G_{k_2})\ln(G_{k_3})(G_{k_1}-G_{k_2}G_{k_3})}(-G_{k_1}+G_{k_2}G_{k_3}), \forall k_1\in\mathbb{I}.
\end{aligned}
\end{equation}
Suppose that  $G$ represents the column vector $(G_1,\dots,G_I)^T$. Let us also denote by  $A_k$, with an abuse of notation, the vector 
$$\left( {\begin{array}{cc}
 0\\       \cdots\\ 1\\ \cdots \\   0  \end{array} } \right),$$
in which the only element that different from $0$ is the $k$-th one.

As discussed about, we convert $F_k$ into $G_k$. This technique allows us to  changes the irreversible system \eqref{BioChemEq1a}-\eqref{BioChemEq1b} into the reversible one $A_{k_1}+A_{k_2}\leftrightarrow A_{k_3}$. Let us now compute the reaction rate functions, which is quite different from mass action kinetics
$$V_{A_{k_1}+A_{k_2} \rightarrow A_{k_3}}(G):=2\mathcal{V}_{k_1,k_2,k_3}\frac{\ln(G_{k_3})-\ln(G_{k_2}G_{k_1})}{\ln(G_{k_1})\ln(G_{k_2})\ln(G_{k_3})(G_{k_3}-G_{k_2}G_{k_1})}{G_{k_1}G_{k_2}},$$
$$V_{A_{k_3} \rightarrow A_{k_1}+A_{k_2}}(G):=2\mathcal{V}_{k_1,k_2,k_3}\frac{\ln(G_{k_3})-\ln(G_{k_2}G_{k_1})}{\ln(G_{k_1})\ln(G_{k_2})\ln(G_{k_3})(G_{k_3}-G_{k_2}G_{k_1})}{G_{k_3}},$$
$$\mathcal{V}_{A_{k_1}+A_{k_2}\leftrightarrow A_{k_3}}:=2\mathcal{V}_{k_1,k_2,k_3}.$$
Otherwise, if $k_1= k_2$ , we write
$$V_{2A_{k_1}\rightarrow A_{k_3}}(G):=\mathcal{V}_{k_1,k_1,k_3}\frac{\ln(G_{k_3})-2\ln(G_{k_1})}{\ln(G_{k_3})\ln(G_{k_1})^2(G_{k_3}-G_{k_1}^2)}{G_{k_1}G_{k_2}},$$
$$V_{A_{k_3} \rightarrow 2A_{k_1}}(G):=\mathcal{V}_{k_1,k_1,k_3}\frac{\ln(G_{k_3})-2\ln(G_{k_1})}{\ln(G_{k_3})\ln(G_{k_1})^2(G_{k_3}-G_{k_1}^2)}{G_{k_3}},$$
$$\mathcal{V}_{2A_{k_1}\leftrightarrow A_{k_3}}:=2\mathcal{V}_{k_1,k_1,k_3}.$$
Using these reaction rate functions,  the   system \eqref{DiscreteQuantumConverted} could be converted into
\begin{eqnarray}\label{EqG}
{\dot{G}}
&=&\mbox{diag}\left( {\begin{array}{cc}
  {(\ln(G_1))^2}\\       \cdots \\   {(\ln(G_I))^2}  \end{array} } \right)\times\\\nonumber
& &\times\sum_{k_1+k_2=k_3}\left[V_{A_{k_1}+A_{k_2} \rightarrow A_{k_3}}(G)-V_{A_{k_3} \rightarrow A_{k_1}+A_{k_2} }(G)\right](A_{k_3}-A_{k_1}-A_{k_2} ).
\end{eqnarray}
Equivalently, we can also write the following equation for the new reversible system  $A_{k_1}+A_{k_2}\leftrightarrow A_{k_3}$
\begin{equation}\label{EqG2}
{\dot{G}}
=\mbox{diag}\left( {\begin{array}{cc}
  {(\ln(G_1))^2}\\       \cdots \\   {(\ln(G_I))^2}  \end{array} } \right)\sum_{y\leftrightarrow y'}\left[V_{y \rightarrow y'}(G)-V_{y' \rightarrow y}(G)\right](y'-y ),
\end{equation}
where $y\leftrightarrow y'$ belongs to the set of reversible reactions
\begin{eqnarray}\label{ChemicalNetworkReversible1a}
A_{k_1}+A_{k_2}&\longleftrightarrow &A_{k_3},
\end{eqnarray}
with $k_1+k_2=k_3$.

Since we have converted \eqref{BioChemEq1a}-\eqref{BioChemEq1b} into the reversible system $A_{k_1}+A_{k_2}\leftrightarrow A_{k_3}$, we can now employ classical techniques for reversible systems, starting with the definitions of the two functionals $\mathcal{R}(G)$ and $\mathcal{S}(G)$ 
\begin{eqnarray}\label{DefineR}\nonumber
&&\mathcal{R}(G)=\\
&=&\mbox{diag}\left( {\begin{array}{cc}
  {(\ln(G_1))^2}\\       \cdots \\   {(\ln(G_I))^2}  \end{array} } \right)\sum_{y \leftrightarrow y'}\left[V_{y \to y'}(G)-V_{y'\to y}(G)\right](y'-y)\\\nonumber
&=&\mbox{diag}\left( {\begin{array}{cc}
  {(\ln(G_1))^2}\\       \cdots \\   {(\ln(G_I))^2}  \end{array} } \right)\sum_{y \leftrightarrow y'}[\mathcal{V}_{y\leftrightarrow y'}G^y -\mathcal{V}_{y\leftrightarrow y'}G^{y'}] \mathcal{H}_{y\leftrightarrow y'}(G)(y'-y),
\end{eqnarray}
and
\begin{eqnarray}\label{DefineRR}\nonumber
\mathcal{S}(G)
&=&\sum_{y \leftrightarrow y'}[\mathcal{V}_{y\leftrightarrow y'}G^y -\mathcal{V}_{y\leftrightarrow y'}G^{y'}]\mathcal{H}_{y\leftrightarrow y'}(G)(y'-y).
\end{eqnarray}
\\ Next, we will follow the techniques introduced in \cite{CraciunFeinberg:2005:MEI} for reversible systems, by computing the Jacobian of $\mathcal{S}$ at the equilibrium point $G^*$, applied to an arbitrary vector $\delta\ne 0$  that belongs to the span of the vectors $y'-y$
\begin{eqnarray}\label{JacR}
\mbox{Jac}(\mathcal{S}(G^*))\delta =\sum_{y \leftrightarrow y'} \mathcal{V}_{y\leftrightarrow y'}(G^*)^y ((y-y')*\delta)\mathcal{H}_{y\leftrightarrow y'}(G^*)(y-y'),
\end{eqnarray}
since $\mathcal{V}_{y\leftrightarrow y'}(G^*)^y -\mathcal{V}_{y\leftrightarrow y'}(G^*)^{y'}=0$
in which the inner product $*$ is defined as
$$y*\delta=\sum_{1}^I \frac{y_k \delta_k}{G_k}.$$
Therefore 
\begin{eqnarray}\label{JacR1}
&& [\mbox{Jac}(\mathcal{S}(G^*))\delta]*\delta =\\\nonumber
& =& \sum_{y \leftrightarrow y'} \mathcal{V}_{y\leftrightarrow y'}(G^*)^y\mathcal{H}_{y\leftrightarrow y'}(G^*)[(y-y')*\delta][(y'-y)*\delta]<0.
\end{eqnarray}
Now, we compute the Jacobian of $\mathcal{R}$ at the equilibrium point $G^*$,
\begin{eqnarray}\nonumber
& &\mbox{Jac}(\mathcal{R}(G^*))\\\nonumber
&=&\mbox{diag}\begin{bmatrix}
\partial_{G_1} {(\ln(G_1^*))^2}\mathcal{S}(G^*)_1\\\cdots     \\    \partial_{G_I} {(\ln(G_I^*))^2}\mathcal{S}(G^*)_I\end{bmatrix}+\mbox{diag}\begin{bmatrix}
  {(\ln(G_I^*))^2}\\\cdots     \\    {(\ln(G_I^*))^2}\end{bmatrix} \mbox{Jac}(\mathcal{S}(G^*))\\\nonumber
&=&\mbox{diag}\begin{bmatrix}
 {(\ln(G_1^*))^2} \\\cdots     \\   {(\ln(G_I^*))^2} \end{bmatrix} \mbox{Jac}(\mathcal{S}(G^*)),
\end{eqnarray}
where the second equality is due to the fact that since $G^*$ is an equilibrium we have that $\mathcal{S}(G^*)=0$.
Since $$\mathfrak{D}:=\mbox{diag}\begin{bmatrix}
  {(\ln(G_1^*))^2} \\\cdots     \\    {(\ln(G_I^*))^2}\end{bmatrix} $$
is a diagonal matrix and $ \mathfrak{J}:=\mbox{Jac}(\mathcal{S}(G^*))$ is negative definite, then $\mathfrak{D}^{1/2}\mathfrak{J}\mathfrak{D}^{1/2}$ is also negative definite with respect to this inner product. Since 
$$\mbox{det}(\mathfrak{D}\mathfrak{J}-\Lambda Id)=\mbox{det}(\mathfrak{D}^{1/2}\mathfrak{J}D^{1/2}-\Lambda Id),~~\forall \Lambda\in\mathbb{R},$$
it follows that $\mathfrak{D}^{1/2}\mathfrak{J}\mathfrak{D}^{1/2}$ and $\mathfrak{D}\mathfrak{J}$ have the same eigenvectors, so $\mathfrak{D}\mathfrak{J}$ is  negative definite. In other words, $\mbox{Jac}(\mathcal{R}(G^*))$ is negative definite.
The exponential rate of convergence
$$\max\{|G_1(t)-G_1^*|,\cdots, |G_I(t)-G_I^*|\}\leq C_1 e^{-C_2 t}.
$$
 then follows from the fact that the Jacobian above is negative definite. This  leads to the conclusion of the theorem.
\end{proof}

\begin{remark}
In the proof above we could have used the Lyapunov function 
\begin{equation}
L(F)=-\Pi_{k=1}^I F_k,
\end{equation}
and all the computations remain the same.
\end{remark}

\bigskip

\section{A reaction network approach to the near resonance equation}\label{NRSec:Q12}
\subsection{The dynamical system associated to the near resonance equation}
Similar as in Section  \ref{Sec:Q12}, let us consider the discrete version of  \eqref{NRQuantumBoltzmann1}, which is described below.
\\\\ Let $\mathcal{L}_R$ denote the lattice of integer points 
$$\mathcal{L}_R=\{p\in\mathbb{Z}^3, |p|<R\}.$$
The discrete version of the near-resonance  equation  \eqref{NRQuantumBoltzmann1} reads 
\begin{equation}\label{NRDiscreteQuantum}\begin{aligned}
\dot{f}_{p_1}=
&\sum_{\substack{p_2,p_3\in\mathcal{L}_R,\\ p_1-p_2-p_3=0, \\ |\omega(p_1)-\omega(p_2)-\omega(p_3)|<\Lambda}}\mathcal{V}_{p_1,p_2,p_3}\left[f_{p_2}f_{p_3}-f_{p_1}(f_{p_2}+f_{p_3})\right]\\
&\quad -2\sum_{\substack{p_2,p_3\in\mathcal{L}_R,\\ p_1+p_2-p_3=0,\\ |\omega(p_1)+\omega(p_2)-\omega(p_3)|<\Lambda}}\mathcal{V}_{p_1,p_2,p_3}\left[f_{p_1}f_{p_2}-f_{p_3}(f_{p_1}+f_{p_2})\right],~~\end{aligned}
\end{equation}
for all $p_1$ in $\mathcal{L}_R$.

Since $\mathcal{L}_R$ contains a finite number of grid points, then there exists a positive number $\Lambda_*$ satisfying
\begin{equation}
\label{LStar}
\Lambda_*=\min_{\substack{p_1,p_2,p_3\in\mathcal{L}_R,\\ p_1-p_2-p_3 = 0, \\ |\omega(p_1)-\omega(p_2)-\omega(p_3)| \neq 0}}|\omega(p_1)-\omega(p_2)-\omega(p_3)|
\end{equation}
See Figure 1 for a geometric illustration of the construction of $\Lambda_*$. 
The red triangle in Figure 1 shows an example of a triple $(p_1, p_2, p_3)$ such that we have $p_1\!-\!p_2\!-\!p_3\!=\!0$ but $|\omega(p_1)-\omega(p_2)-\omega(p_3)| \neq 0$.
Moreover, we also define
\begin{equation}
\label{LStar2}
\Lambda^*=\min_{\substack{p\in\mathcal{L}_R} \setminus \{0\} } | p |.
\end{equation}

\subsection{Convergence to equilibrium}

\begin{theorem}\label{NRTheoremQ12}
The solution of the discrete near resonance equation \eqref{NRDiscreteQuantum} has different behaviors for different values of the resonance broadening frequency $\Lambda$.

\begin{itemize}
\item If $\Lambda<\Lambda_*$, where $\Lambda_*$ is defined in \eqref{LStar}: For any positive initial condition, the solution $$f(t)=(f_p(t))_{p\in\mathcal{L}_R}$$ of the discrete wave turbulence equation \eqref{NRDiscreteQuantum} converges to an equilibrium state $f^*=(f_p^*)_{p\in\mathcal{L}_R}$. For each ray $\{kP_0\}_{k\geq 1}$ there exists a positive constant $\rho({P_0})$ such that if $p=kP_0$ then $$f^*_p=\frac{1}{{k\rho({P_0})}}.$$ Moreover, the solution $f(t)$ of \eqref{NRDiscreteQuantum} converges to $f^*$ exponentially fast, in the following sense: there exist positive constants $C_1$, $C_2$ such that 
\begin{equation}\label{Rate1}
\max_{p\in\mathcal{L}_R}|f_p(t)-f_p^*|<C_1e^{-C_2t}.
\end{equation}

The reason for \eqref{Rate1} to hold is that when $\Lambda<\Lambda_*$, the resonance broadening equation becomes the exact resonance one.
\item If $2\Lambda^*\le \Lambda$, where $\Lambda^*$ is defined in \eqref{LStar2}: For any positive initial condition, the solution $$f(t)=(f_p(t))_{p\in\mathcal{L}_R}$$ of the discrete wave turbulence equation \eqref{NRDiscreteQuantum} exits any compact set $K\subset (0,\infty)^{|\mathcal{L}_R|}$ as $t$ tends to infinity, and moreover we have $\displaystyle \liminf_{t \to \infty} || f(t) || = \infty.$
\end{itemize}

\end{theorem}
\begin{proof}
\bigskip

{\it Step 1: The Lyapunov function.}
We define the  function
\begin{equation}\label{NRLyapunov}
L(f)=-\sum_{p\in \mathcal{L}_R}^I \log (f_{p}),
\end{equation}
and we will show that $L$ is a Lyapunov function for  \eqref{NRDiscreteQuantum}.
Similar as the proof of Theorem \ref{NRTheoremQ12}, we also have
\begin{equation}\label{LyapunovInequality}
\begin{aligned}
& \ \dot{f}\cdot\nabla L = \\
 =  &- \sum_{p_2+p_3=p_1, ||p_2|+|p_3|-|p_1||\le \Lambda}\mathcal{V}_{p_1,p_2,p_3}f_{p_1}f_{p_2}f_{p_3}\left[\frac{1}{f_{p_1}}-\frac{1}{f_{p_2}}-\frac{1}{f_{p_3}}\right]^2\\
 \leq  & \  0,
\end{aligned}
\end{equation}
Note that the above inequality is strict unless
\begin{equation}\label{NRC12equilibrium}
\frac{1}{f_{p_1}}
=\frac{1}{f_{p_2}}+\frac{1}{f_{p_3}},\end{equation}
for all  $p_1=p_2+p_3$ and $||p_2|+|p_3|-|p_1||\le \Lambda$.\\

{\it Step 2: The two cases of $\Lambda$.}

{\bf Case 1: $\Lambda<\Lambda_*$}.

Since $\Lambda<\Lambda_*$, the system  $p_1=p_2+p_3$ and $||p_2|+|p_3|-|p_1||\le \Lambda<\Lambda_*$ becomes exactly  the system  $p_1=p_2+p_3$ and $||p_2|+|p_3|-|p_1||=0.$
Therefore, equation \eqref{NRC12equilibrium}  implies that at equilibrium   $f_k^*=\rho k$, for some positive constant $\rho$. 
By the conservation relation
$$\sum_{p\in \mathcal{L}_R}p{f_p}=\sum_{p\in \mathcal{L}_R}p{f_p^*},$$
we deduce that $\rho$ is unique and $f^*$ is the only equilibrium point that satisfies the same conservation relation as the initial condition. The same argument as in Theorem \ref{TheoremQ12} can be applied and the conclusion of Theorem \ref{NRTheoremQ12} follows.

\begin{figure}[!h] \begin{center}
  \includegraphics[width=0.5\textwidth, height=0.3\textheight]{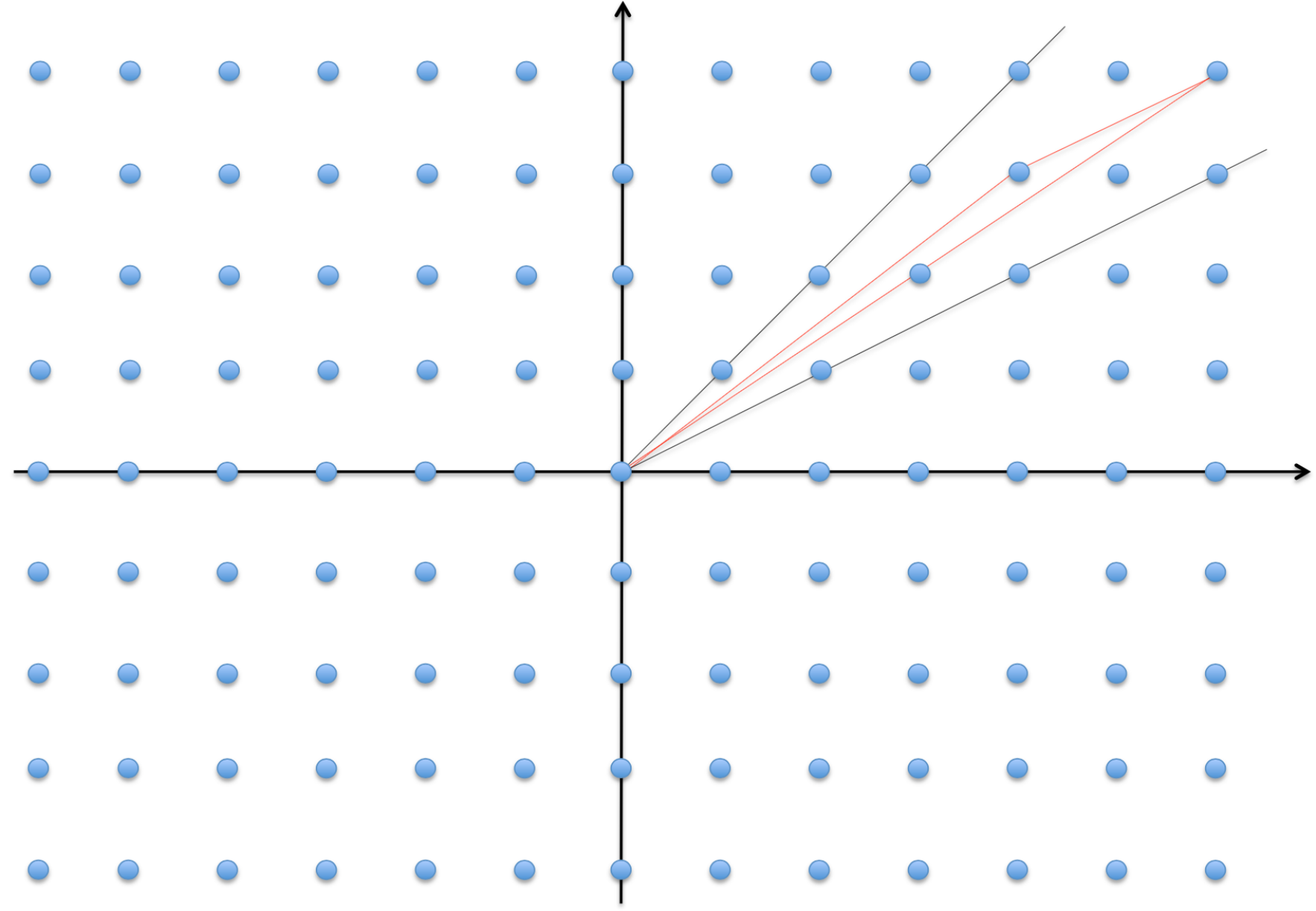}
\caption{The system of equations \eqref{DiscreteQuantum} can be decoupled alongs rays, as shown in \eqref{DiscreteQuantum1DVersion2}. Two examples of such rays are shown in black above. The red triangle illustrates the definition of $\Lambda_*$ in (\ref{LStar}). }
\label{Fig2}\end{center}
\end{figure}

{\bf Case 2: $2\Lambda^*\le \Lambda$}.

Define $Z_p=\frac{1}{F_p}$. The equilibrium set of \eqref{LyapunovInequality} satisfies the following system of linear equations. 

\begin{equation}\label{EquilibriumSystem}
{Z_{p_1}}
={Z_{p_2}}+{Z_{p_3}},\end{equation}
for all  $p_1=p_2+p_3$ and $||p_2|+|p_3|-|p_1||\le \Lambda$.
We will show that the system \eqref{EquilibriumSystem} has no solution by contradiction.

The system \eqref{EquilibriumSystem} contains the exact resonance one as a subsystem 
\begin{equation}\label{ExactEquilibriumSystem}
{Z_{p_1}}
={Z_{p_2}}+{Z_{p_3}},\end{equation}
for all  $p_1=p_2+p_3$ and $||p_2|+|p_3|-|p_1||= 0$. 

For each ray $p$, we denote by $e_p$  the closet point of $\mathcal{L}^R$ to the origin on this ray. Suppose that $p$ is chosen such that $|e_p|=\Lambda^*$.  Let us consider all of the points of $\mathcal{L}^R$, which is a set of the form $\{e_p, 2e_q, 3e_q,\cdots, Je_q\}.$ On this set, the exact resonance equation  \eqref{ExactEquilibriumSystem} is the classical one
$$Z_{(i+j)e_p}=Z_{ie_p}+Z_{je_p}$$
and has a unique solution
$$Z_p=\lambda|p|,$$
where $\lambda$ is a positive constant. 

By a similar argument, let us consider the ray $-e_p$; then $-e_p$ is also the closest point of $\mathcal{L}^R$ to the origin on this ray. Thus, we know
$$Z_q=\lambda'|q|,$$
where $\lambda'$ is a positive constant, for all $q$ on the ray of  $-e_p$.

Consider the vectors $p_1=3e_p$, $p_2=4e_p$ and $p_3=-e_p$, then $p_1=p_2+p_3$ and $||p_2|+|p_3|-|p_1||=2\Lambda_*\le \Lambda$. Therefore for these choices
$$Z_{p_1}=Z_{p_2}+Z_{p_3},$$
which is equivalent to
 $$\lambda3|e_p|=\lambda4|e_p|+\lambda'|e_p|.$$
 This leads to a contradiction.

Since the system \eqref{EquilibriumSystem} has no solution, for any compact set $K\subset (0,\infty)^{|\mathcal{L}_R|}$,
  there exists $\Delta_K >0$ such that 
\begin{equation}\label{LyapunovInequality2}
\begin{aligned}
& \ \dot{f}\cdot\nabla L = \\
 =  &- \sum_{p_2+p_3=p_1, ||p_2|+|p_3|-|p_1||\le \Lambda}\mathcal{V}_{p_1,p_2,p_3}f_{p_1}f_{p_2}f_{p_3}\left[\frac{1}{f_{p_1}}-\frac{1}{f_{p_2}}-\frac{1}{f_{p_3}}\right]^2\\
 \leq  & \  -\Delta_K.
\end{aligned}
\end{equation}
Therefore $f$ exits any compact set $K\subset (0,\infty)^{|\mathcal{L}_R|}$ as $t$ tends to infinity.
In particular, consider the compact sets of the form 
\begin{equation} \label{compact_special}
\{ x \in (0,\infty)^{|\mathcal{L}_R|} \  | \   \prod_{p\in \mathcal{L}_R}  x_{p} \ge \eps_0 \textrm{ and } || x || \le M\}.
\end{equation}
Note that the set 
\begin{equation} \label{level_set}
\{ x \in (0,\infty)^{|\mathcal{L}_R|} \  | \   \prod_{p\in \mathcal{L}_R}  x_{p} = \eps_0 \}
\end{equation}
is a {\em level set} of the Lyapunov function $L(f)$ (see also Remark \ref{Lyap_prod} below).

Consider now a solution $f(t)$ of equation \eqref{NRDiscreteQuantum} and some $\eps_0 > 0$ such that for $f^0=f(0)$ we have $\prod_{p\in \mathcal{L}_R}  f^0_{p} > \eps_0$.

Then, since $f(t)$ must exit the compact set given by (\ref{compact_special}) for any $M>0$, and the Lyapunov function property prevents it from exiting through the boundary of the form (\ref{level_set}), it follows that $\displaystyle \liminf_{t \to \infty} || f(t) || = \infty$.
\end{proof}

\bigskip

\begin{remark} \label{Lyap_prod}
In the proof of Theorem \ref{NRTheoremQ12} we could also have used the Lyapunov function 
\begin{equation}
L(f) = - \prod_{p\in \mathcal{L}_R}  f_{p},
\end{equation}
and all the computations remain the same.
\end{remark}

\section{Perspectives}

The  explicit expression for the Kolmogorov-type spectrum
of acoustic turbulence of equation \eqref{QuantumBoltzmann1} is (cf. \cite{l1997statistical})
$$f(p)=|p|^{-3/2}.$$
An open question is if this is also a spectrum of the discrete system. In our work, we prove that $|p|^{-1}$
is a global attractor for the discrete system. In a future work, we will try to compare these two solutions. We suspect that it may be not possible to find a Lyapunov function for $|p|^{-3/2}$ despite the fact that it is possible to find a Lyapunov function for $|p|^{-1}$, as shown above. Moreover, we will also plan to study the behavior of the system when $\Lambda_*\le \Lambda\le 2\Lambda^*$.  

~~ \\{\bf Acknowledgements.} G. Craciun was supported by NSF grants DMS-1412643 and DMS-1816238. M.-B. Tran was partially supported by NSF Grant DMS-1814149, NSF Grant DMS-1854453, SMU URC Grant 2020 and Linking Fellowship of Dedman College of Humanities and Sciences.

\def\cprime{$'$} \def\cprime{$'$} \def\cprime{$'$} \def\cprime{$'$}
  \def\cprime{$'$} \def\cprime{$'$} \def\cprime{$'$} \def\cprime{$'$}

\end{document}